\documentclass[12pt,epsfig]{amsart}
\usepackage{amsthm}
\usepackage{amssymb}
\usepackage{amsfonts}
\usepackage{amsmath}

\usepackage{braket}
\usepackage{dsfont}
\usepackage{relsize}

\usepackage{latexsym}
\usepackage{tikz}
\usepackage{braket}

\usepackage{amsmath,amscd,amssymb}
\usepackage{graphics}
 2

\usepackage[utf8]{inputenc}
\usepackage[T1]{fontenc}

\include{plaquettemacro}
\usepackage{hyperref}

\theoremstyle{plain}
\newtheorem{thm}{Theorem}[section]
\newtheorem{lem}{Lemma}[section]

\theoremstyle{definition}

\newtheorem{prop}{Proposition}[section]
\theoremstyle{remark}

\begin{document}

\title{ Coadjoint orbits and K\"ahler structure: \\examples from coherent states }

\maketitle

\begin{center}
\author{Rukmini Dey \\
I.C.T.S.-T.I.F.R.,   S\lowercase{ivakote}, H\lowercase{esaraghatta} H\lowercase{obli}  \\
 B\lowercase{angalore} 560089, I\lowercase{ndia}\\
\lowercase{rukmini@icts.res.in}\\
Joseph Samuel \\
I.C.T.S.-T.I.F.R.,  S\lowercase{ivakote}, H\lowercase{esaraghatta} H\lowercase{obli}\\
 B\lowercase{angalore 560089,} I\lowercase{ndia } \\
 and RRI,  S\lowercase{adashivanagar},  B\lowercase{angalore} 560080, I\lowercase{ndia}\\
\lowercase{sam@icts.res.in and sam@rri.res.in}\\
Rithwik  S. Vidyarthi\\
  D\lowercase{epartment of} M\lowercase{athematics}\\
C212 W\lowercase{ells} H\lowercase{all}\\
M\lowercase{ichigan} S\lowercase{tate} U\lowercase{niversity}\\
619 R\lowercase{ed} C\lowercase{edar} R\lowercase{oad}\\
E\lowercase{ast} L\lowercase{ansing}, MI 48824\\
\lowercase{rthwkvidyarthi@gmail.com}}
\end{center}

\begin{abstract}

Do co-adjoint orbits of 
Lie groups support a K\"{a}hler structure? We study this question from
a point of view derived from coherent states. We examine three examples
of Lie groups: the Weyl-Heisenberg group, $\mathrm{SU(2)}$ and 
$\mathrm{SU(1,1)}$. In cases, where the orbits admit a  
K\"{a}hler structure,
we show that  coherent states give us 
a K\"{a}hler embedding of the orbit into projective Hilbert space.
In contrast, squeezed states, (which like coherent states, also saturate the uncertainty
bound) only give us a symplectic embedding. We also study geometric quantisation  of  
the co-adjoint orbits of the group $\mathrm{SUT(2,\mathbb{R})}$ of real, special, 
upper triangular matrices in two dimensions. We glean some general insights from these examples.
Our presentation is semi-expository and accessible to physicists. 

\end{abstract}

\vskip 5pt
Keywords: Coherent states,  Squeezed States, Coadjoint orbits, Toda system

\section{Introduction}

Coadjoint orbits of Lie groups have a symplectic structure ~\cite{DWI}.  
It is natural to ask: do they have a K\"ahler structure? 
This question has been raised before ~\cite{V} in a different context.
In a recent thesis ~\cite{V}, Villa studied the coadjoint orbits of
semisimple Lie groups.
The motivation was to understand the geometry of the convex hull of coadjoint
orbits. In Theorem 2.27,  ~\cite{V} gives a necessary and sufficient
condition for the existence of a K\"{a}hler structure {\it{for semi-simple Lie
groups}}. 
In this semi-expository paper we study this question from the point of 
view of Perelomov coherent states in three examples. These are 
the Weyl-Heisenberg group,  $\mathrm{SU(2)}$ and $\mathrm{SU(1,1)}$.  
In each of these cases we also discuss squeezed states using a 
definition of squeezing that is motivated by
the Robertson-Schr\"{o}dinger uncertainty relations~\cite{Robertson,schrod}. 
Our paper also includes a discussion of
the group of special upper triangular matrices 
(non-zero entries only on the diagonal and above)
and  geometric quantization of its coadjoint orbits.

Coherent states were first studied by Schr$\ddot{\rm{o}}$dinger in an effort to
construct quantum states with classical-like behaviour \cite{schrodinger,fock,
sudarshan,glauber,klauder,optics}. He searched
for states which saturate the inequality expressed by the Heisenberg
uncertainty principle
\begin{equation}
\Delta q \Delta p\ge \hbar/2.
\label{uncertainty}
\end{equation}
These states are called minimum uncertainty states. By minimising the quantum 
uncertainty, one finds quantum states that are the closest
one can get to classical states. States of a classical system 
are described as points in phase space, which is a symplectic manifold. 
Coherent states can be thought of as localised around these points. 
They have equal uncertainty in position and momentum. These states are
called Schr\"{o}dinger coherent states (SCS) and are central to the physics
of harmonic oscillators and optics.

However, requiring minimum uncertainty is 
an incomplete characterisation of coherent states.
There exist other states which also saturate the Heisenberg bound. 
These states are called squeezed states \cite{Schnabel} and they play an
important role in physics. 
Squeezed states have a larger uncertainty in one variable (either position or momentum)
and a smaller uncertainty in the conjugate variable. From the viewpoint
of minimising uncertainty, squeezed states are just as good as coherent
states. Not only do squeezed and coherent states saturate the Heisenberg
uncertainty relation, but in  fact, they both saturate a slightly stronger
version, the Robertson-
Schr\"odinger inequality\cite{Robertson,schrod}.  We will elaborate on this in Appendix 1.

Squeezed states are extremely important in quantum optics
and quantum metrology.
Squeezed states of light are used, 
for example, in the LIGO detector to give an accurate determination
of position, while sacrificing accuracy in the conjugate variable,
momentum. Although, both saturate the uncertainty bound, 
squeezed states are physically quite different from coherent states.

Coherent states have been studied mathematically
from advanced group theoretic points of view\cite{perelomov}, 
using co-adjoint orbits of Lie Groups. The symplectic structure on the 
co-adjoint orbit 
gives the orbit the character of a classical phase space. Coherent states then give us an embedding of 
the classical phase space into the ray space of quantum mechanics. The ray space of quantum mechanics
admits a K\"{a}hler structure. Can one pull back this structure to the coadjoint orbit? 
Does the pullback agree with the K\"{a}hler structure (if it exists) on the coadjoint orbit?
Can coherent states be used to {\it define} a K\"{a}hler structure
on coadjoint orbits? These are the questions which motivate 
the present study. 
 We will show  in three pertinent examples 
that coherent states give us a K\"ahler embedding
of the classical phase space into the ray space of quantum mechanics,
while squeezed states give us only a symplectic embedding. This property
sets apart coherent states from the other minimum uncertainty 
states.

We  show from first principles the  fact  that the Weyl-Heisenberg coherent states, $\mathrm{SU(2)}$ coherent states and $\mathrm{SU(1,1)}$ coherent states embed the respective  coadjoint orbits into 
projective Hilbert space
such that the pull back of the Fubini Study form is K\"ahler,  thus illustrating a general result in \cite{{DIL}, {LM}, {Od}}.
The coadjoint orbits of 
interest are ${\mathbb C}$,  $ \mathrm{SU(2)}/\mathrm{U(1)} 
\equiv \mathrm{S^2}$ and $\mathrm{SU(1,1)}/\mathrm{U(1)} 
\equiv \mathrm{H}$ (upper half plane),  respectively. 
Next we  show that the pull back by squeezed states,  however,  
yields only a symplectic structure on the coadjoint orbit.

Our exposition includes a detailed study of the 
geometric quantization of the coadjoint orbits of the special 
upper triangular matrices. 
A  coadjoint orbit of  $\mathrm{SUT}$ is $\mathrm{SUT}^{+}$ \footnote{
In fact, this orbit is also a group, and represents the $AN$ part of the $KAN$ or
Iwasawa decomposition of $\mathrm{SL(2,\mathbb{R})}$.} which is intimately connected to  the  $2$-dimensional Toda system. 
In  ~\cite{Ad} Adler showed that a finite $n$-dimensional Toda system has a coadjoint 
orbit description of the group of lower triangular matrices of non-zero diagonal. 
In fact, one can restrict the  action to that of lower triangular matrices of 
determinant $1$ and positive diagonal elements. The orbit is homeomorphic to ${\mathbb{R}}_+^{n-1} \times {\mathbb{R}}^{n-1}$, 
just described by $a_i >0$, $i = 1,..., n-1$ and $ b_i$, $i =1,...,n$  such that $b_1+ b_2 + ...+ b_n =c$, $c$ a constant. 
In this paper we deal with the case $n=2$.  
In \cite{DG}, the authors geometrically quantized this system and studied the coherent states.
Coherent and squeezed states pulled back from projective spaces
have been considered in \cite{DGh} in a somewhat more general context.

The paper is structured as follows: In section 2,  (and in appendix 1) we introduce
coherent and squeezed states and describe how Schr\"odinger 
coherent states emerge
from geometric quantisation of co-adjoint orbits of the Weyl-Heisenberg group. 
In section 3, we define the ray space of quantum mechanics and 
describe the K\"ahler structure of the ray space.  In section 4 we  show that
the Schr\"odinger coherent states give us a K\"ahler embedding of the complex
plane into the ray space. and we show a similar embedding of the unit disc ($\mathrm{SU(1,1)}$ coherent states) and the sphere ($\mathrm{SU(2)}$ coherent states).  We show squeezed states give only a symplectic embedding.   In section 5, we study a coadjoint orbit
of the group of special upper triangular matrices. The phase space is the upper  half plane.  Section 6 is a concluding discussion.
An appendix (section 7)
describes the origin of our definition of squeezed states  
and another appendix treats the Berezin 
quantisation of the upper half plane 
which  is already known for the unit disc.

\section{Squeezed and Coherent States}
{\it Schr\"odinger coherent states:}
Schr\"odinger coherent states are defined as solutions to the eigenvalue
equation
\begin{equation}
(\hat{q}+i \hat{p})\ket{\alpha}=\alpha \ket{\alpha},
\label{coherent}
\end{equation}
where the eigenvalue $\alpha$ is a complex number.
These states saturate the uncertainty bound eq.(\ref{uncertainty}). 
The eigenvalue equation is easily solved to give an expression for coherent
states in terms of the oscillator eigenstates $\ket{n}$.
\begin{equation}
\ket{\alpha}=\exp[{-|\alpha|^2/2}]
\mathlarger{\sum}_{n=0}^{\infty}\frac{\alpha^n}{\sqrt{n!}} \ket{n}, 
\label{coherexpression}
\end{equation}
$\ket{n}$ are orthonormal 
eigenstates of the Hermitian operator $\hat{H}=(\hat{q}^2+\hat{p}^2)/2$.

{\it Squeezed states:}
As mentioned above, coherent states are not the
only states which saturate
the uncertainty inequality.
Squeezed states, defined by

\begin{equation}\label{squeezeddefinition}
\left(\hat{q}\lambda+i \frac{\hat{p}}{\lambda}\right)\ket{\alpha, \lambda}=\alpha \ket{\alpha, \lambda},
\end{equation}
also share this property.  
We describe  the orgin of this defintion of squeezed states in appendix 1.
Squeezed
states are anisotropic in phase space being ``squeezed"
in one direction and expanded
in the other. This squeezing operation in phase space 
preserves the symplectic structure but alters the complex structure
of the plane.  

Mathematically, Schr\"odinger
coherent states emerge naturally
from the Weyl-Heisenberg group ${\mathcal W}$,
which is the exponential of
the nilpotent W-H Lie algebra, generated
by $e_1,e_2,e_3$ 
with the only nonzero commutation relation being
\begin{equation}
[e_1,e_2]=e_3.
\label{whcommutation}
\end{equation}
The mathematical theory of coherent states starts with 
an irreducible, unitary 
representation of the W-H group
in a Hilbert space \cite{perelomov}. Taking
a fiducial vector $\ket{\psi_0}$ satisfying
\begin{equation}
(\hat{q}+i \hat{p})\ket{\alpha}=0,
\label{fiducial}
\end{equation}
we have
\begin{equation}
\ket{\alpha}=D(\alpha)\ket{\psi_0},
\label{alphadef}
\end{equation}
where the complex number $\alpha$ determines
$D(\alpha)=\exp{[i\hat{q}\alpha_1+i\hat{p}\alpha_2]}$, and the group commutation
relations ensure that (\ref{coherent})  is satisfied \cite{perelomov}.
The complex number $\alpha$ (or equivalently the two 
real numbers ($\alpha_1,\alpha_2$), where $\alpha = \alpha_1 + i \alpha_2$, 
parametrises the points of 
a two dimensional co-adjoint orbit of ${\mathcal W}$. The co-adjoint 
orbit $\Gamma$ has a natural symplectic structure from the 
Kirillov-Konstant-
Souriau construction, ~\cite{W}.
$\Gamma$ is in fact the classical phase space of the Harmonic oscillator,
identified with the complex plane. $\Gamma$ admits a K\"ahler structure,
with the symplectic, Riemannian and complex structures coexisting
compatibly.

\section{Ray Space as a K\"ahler Manifold}

The states of a quantum system are described projectively as rays in 
a Hilbert space ${\mathcal H}$. Let us consider the space of normalised states
${\mathcal N}=\{ \psi \in {\mathcal H}:\|\psi\|^2 =1\} $.
We can regard the 
Ray space as 
elements of ${\mathcal N}$,  modulo
a overall phase 
\begin{equation}
{\mathcal R}={\mathcal N}/\sim,
\label{rayspace}
\end{equation}
where $\ket{\psi}\sim\ket{\psi'}$ if $\ket{\psi}=\exp{i \gamma}\ket{\psi'}$. An equivalent definition of ${\mathcal R}$ is to view it as the space of one dimensional projections on ${\mathcal H}$. These can
be written as $\rho=|\psi\rangle \langle \psi|$. $\rho$ is Hermitean, 
$\rho^{\dagger}=\rho$, a projection operator $\rho^2=\rho$ and normalised 
$\mathrm{ Tr} (\rho)=1$. In finite dimensional quantum systems such as 
occur physically in spin systems, the ray space is $\mathbb{C}P^n$. In two of the systems
we deal with in this paper, the classical phase space is non-compact
and has infinite symplectic volume. As a result the Hilbert space ${\mathcal H}$
as well as the ray space ${\mathcal R}$ are infinite dimensional. We can 
describe the ray space as ${\mathbb C}P^n \;$ 
or ${\mathbb C}P^{\infty}$ depending on whether the dimension is finite or infinite.

The ray space
${\mathcal R}$ is naturally endowed with a metrical structure,  
(with distances $\delta$ 
determined by shortest geodesics of the Fubini-Study metric)
\begin{equation}
\|\langle\psi_1\rvert\psi_2\rangle\|=\cos{\delta/2},
\label{fubinistudy}
\end{equation}
where $\delta$ ranges from $0$ to $\pi$ and gives 
the distance between rays. $\delta$ is directly 
measurable in a laboratory as a transition probability.  

The ray space also has two more structures inherited from ${\mathcal H}$:
a symplectic structure and a complex structure. The symplectic
structure on the ray space ${\mathcal R}$ is described by  a closed,
non-degenerate two form, the curvature of the universal $\mathrm{U(1)}$ connection on 
the bundle:
${\mathcal N}\rightarrow {\mathcal R}$.
The symplectic structure on 
${\mathcal R}$ has been interpreted as a geometric phase \cite{WiSh}, which also is amenable
to experiments.  The complex structure is less evident in laboratory terms
but follows mathematically from the complex structure on ${\mathcal H}$. 
Physically, the complex structure is crucial to a discussion of time reversal
symmetry, which acts by conjugation on the complex structure. 
The three structures - metrical, symplectic and complex- give us a K\"ahler
structure on the ray space. 

The equation (\ref{coherexpression}) 
defining a coherent state gives us an embedding
of the classical phase phase into the ray space ${\mathcal R}$. Each 
complex number
$\alpha \in \Gamma$, determines an unique nonzero 
element $\ket{\alpha}\in{\mathcal H}$, which projects
down to the ray space $[\ket{\alpha}]$. (The square brackets 
refer to the equivalence class of $\ket{\alpha}$.) This can also be expressed
as a projection operator $\rho=\ket{\alpha}\bra{\alpha}$)    
\begin{gather}
\phi: \Gamma\rightarrow {\mathcal R}\notag\\
\phi(\alpha)= \rho=\ket{\alpha}\bra{\alpha}.
\label{embedd}
\end{gather}

 It is known  that in some examples, 
the coherent states provide us with a K\"ahler embedding
of a coadjoint orbit into the projective Hilbert space ${\mathbb C}P^n$ or ${\mathbb C}P^{\infty}$, \cite{DIL}, \cite{LM},  \cite{Od}. 
In this paper, we investigate the nature of this embedding for coherent
and squeezed states in three examples. 

It is well known  that the classical phase space 
$\Gamma$ of the oscillator is a K\"ahler manifold. As we remarked
above, the ray space ${\mathcal  R}$ has
a K\"ahler structure
which gives mutually compatible
symplectic, complex and metrical
structures. Given the embedding (\ref{embedd}) it is reasonable to 
ask if the pullback of the structures on ${\mathcal R}$ agrees with
the corresponding structure on $\Gamma$. We first discuss this
question for the Schr\"odinger coherent states, 
then the coherent states of $\mathrm{SU(2)}$ and $\mathrm{SU(1,1)}$.

Consider two tangent vectors $\dot{\alpha}$ and $\alpha'$ on the 
complex plane $\Gamma$. The symplectic and metrical structures on the $\alpha$
plane
are given respectively by 
\begin{equation}
\omega(\alpha',\dot{\alpha})= {\rm Im} \left(\alpha^{'*}\dot{\alpha} \right)
\label{symp}
\end{equation}
and 
\begin{equation}
g(\alpha',\dot{\alpha})= {\rm Re} \left(\alpha^{'*}\dot{\alpha} \right).
\label{metric}
\end{equation}

We can evaluate the push-forward of $\alpha'$ and $\dot{\alpha}$ to the ray space 
and use the definition 
(\ref{coherexpression}) of $\ket{\alpha}$, 
the metric and symplectic structure 
on $\mathcal R$, to  find  the real and imaginary parts
of 
\begin{equation}
\bra{\alpha}' P \dot{\ket{\alpha}},
\label{rayspacestructure}
\end{equation}
(where $P$ is the projector $\mathds{1}\ -\ket{\alpha}\bra{\alpha}$ orthogonal
to the rays of ${\mathcal H}$) 

A calculation given in detail in section 4
with (\ref{coherexpression}) shows that this works out 
to the symplectic (\ref{symp})  and metrical structure  (\ref{metric})
on $\Gamma$.
Thus Schr$\ddot{\rm{o}}$dinger coherent states
have the property that the pull back of the K\"ahler 
structure on ${\mathcal R}$  to $\Gamma$  agrees
with naturally occurring
structures on $\Gamma$: $\phi$ gives us a K\"ahler embedding of $\Gamma$ into
${\mathcal R}$.

This property distinguishes
coherent states from squeezed states. Squeezed states are got by composing the map (\ref{embedd}) with 
the transformation $z\rightarrow \lambda 
(z+\Bar{z})/2+\lambda^{-1} (z-\Bar{z})/2$. This transformation
is not complex analytic for $\lambda\ne 1$ and does not respect the complex 
structure of $\Gamma$. However it does preserve the symplectic form.
Squeezed states only give us a symplectic embedding 
but not a K\"ahler embedding.

\section{K\"ahler embedding of the 
classical phase space into projective Hilbert space.}
 In this section we explicitly show in examples that the coherent 
states provide us with a K\"ahler embedding
of a coadjoint orbit into the projective Hilbert space  ${\mathbb C}P^{\infty}$, 
which illustrates a  general result in   \cite{{DIL}, {LM}, {Od}}. In each case we also show that squeezed states only give a symplectic embedding.

\subsection{ K\"ahler embedding of the 
classical phase space: Schr$\ddot{\rm{\bf{o}}}$dinger Coherent states}
The coherent state on the complex plane is given by:
$$\lvert\alpha\rangle=\exp\left(\frac{-\|\alpha\|^2}{2}\right)\sum_0^\infty \frac{\alpha^n}{\sqrt{n!}}\lvert n\rangle.$$
Let $\alpha(s)$ be a curve in the complex plane. We have the following tangent vector on $\mathbb{C}$:
$$\alpha'=\frac{d\alpha}{ds}.$$
This defines a curve $\lvert\alpha(s)\rangle \in  \mathcal{N}$. We can project this down to the ray space to obtain a curve $[\lvert\alpha(s)\rangle] \in \mathcal{R}$.\newline
The push forward of $\alpha'$ to $\mathcal{N}$ is $\frac{d}{ds}\lvert\alpha(s)\rangle$. Composing with $$\pi:\mathcal{H}-\{0\} \to \mathcal{R}$$
we have $$\left(\frac{d}{ds}\lvert\alpha(s)\rangle\right)^\perp=P\left(\frac{d}{ds}\lvert\alpha(s)\rangle\right),$$
where $P=\mathds{1}-\lvert\alpha(s)\rangle\langle\alpha(s)\lvert$ is the projector perpendicular to $\lvert\alpha\rangle$.
\begin{equation*}
    \begin{split}
        \frac{d}{ds}\lvert\alpha(s)\rangle&=\frac{d}{ds}\left(e^{-\frac{\|\alpha\|^2}{2}}\sum_0^\infty \frac{\alpha^n}{\sqrt{n!}}\ket{n}\right)\\
        &=\left(\frac{d}{ds}\left(-\frac{\|\alpha\|^2}{2}\right)\right)\lvert\alpha\rangle+e^{-\frac{\|\alpha\|^2}{2}}\alpha'\sum_{0}^{\infty}\frac{n\alpha^{n-1}}{\sqrt{n!}}\ket{n}.\\
    \end{split}
\end{equation*}
On projection the first term, vanishes. Letting $\ket{b}=\frac{d}{ds}\lvert\alpha(s)\rangle$, we have
\begin{equation}
P \ket{b}=P\left(\frac{d}{ds}\lvert\alpha(s)\rangle\right)=P\left(e^{-\frac{\|\alpha\|^2}{2}}\alpha'\sum_{0}^{\infty}\frac{n\alpha^{n-1}}{\sqrt{n!}}\ket{n}\right).
\label{firstequation}
\end{equation}
Similarly, we set $\ket{b}=\frac{d}{dt}\lvert\alpha(t)\rangle$, to describe another tangent vector in the $\alpha$ plane. This leads to a similar expression as 
(\ref{firstequation}) with $\alpha'$ replaced by $\dot{\alpha}$. 
We need to calculate:
\begin{equation*}
    \begin{split}
        \bra{a} P \ket{b}&=\langle a\rvert(\mathds{1} -\lvert\alpha\rangle\langle\alpha\rvert)\lvert b\rangle\\
        &=\langle a\lvert b\rangle-\langle a\lvert\alpha\rangle\langle\alpha\lvert b\rangle.\\
    \end{split}
\end{equation*}
Carrying out these calculations we obtain:\newline
\begin{flalign*}
\langle a\rvert b\rangle&=\overline{\alpha'}\dot{\alpha}e^{-\|\alpha\|^2}\sum_{m,n=0}^{\infty} \frac{mn{\overline{\alpha}}^{m-1}\alpha^{n-1}}{\sqrt{m!n!}}\langle m\rvert n\rangle &\\
        &=\overline{\alpha'}\dot{\alpha}e^{-\|\alpha\|^2}\sum_{n=0}^\infty \frac{n^2(\overline{\alpha}\alpha)^{n-1}}{n!}&\\
        &=\overline{\alpha'}\dot{\alpha}e^{-\|\alpha\|^2}\left[\overline{\alpha}\alpha e^{\|\alpha\|^2}+e^{\|\alpha\|^2}\right]&\\
        &=\overline{\alpha'}\dot{\alpha}(1+\overline{\alpha}\alpha),&
\end{flalign*}

\begin{flalign*}
\langle\alpha\rvert b\rangle&=\dot{\alpha}e^{-\|\alpha\|^2}\sum_{m,n=0}^\infty\frac{\overline{\alpha}^m}{\sqrt{m!}}\frac{n\alpha^{n-1}}{\sqrt{n!}}\langle m\rvert n \rangle &\\
        &=\dot{\alpha}e^{-\|\alpha\|^2}\sum_{n=0}^\infty\frac{n\overline{\alpha}^n \alpha^{n-1}}{n!}&\\
        &=\dot{\alpha}\overline{\alpha}.&
\end{flalign*}
Similarly,
$\langle a\rvert\alpha\rangle =\alpha\overline{\alpha'}.$

Substituting these values in the equation above we obtain:
$$\bra{a}P\ket{b}=\overline{\alpha'}\Dot{\alpha},$$
which exactly agrees with the Hermitean metric on $\mathbb{C}$.

\subsubsection{ \bf{Symplectic embedding of the 
classical phase space of the harmonic oscillator  by  squeezed  states:}}

To briefly recapitulate the main formulae,
recall that $a = \frac{\left(\hat{q} + i \hat{p}\right)}{\sqrt{2}}$ and $a^{+}  = \frac{\left(\hat{q} - i \hat{p}\right)}{\sqrt{2}}.$ 
Coherent states are defined by the formula 
$a \lvert \alpha \rangle = \alpha \lvert \alpha \rangle$. 
The fiducial (or `vacuum') state  
$\lvert 0 \rangle $ is defined by $a \lvert 0 \rangle =0$,
and the displacement operator  by 
$D(\alpha) = \mathrm{exp}( \alpha a^{+} - \bar{\alpha} a)$.
The displacement operator $D(\alpha)$ satisfies 
$D^{-1}(\alpha) a D(\alpha) = a + \alpha$, so that
$ \lvert \alpha \rangle = D(\alpha) \lvert 0 \rangle $.

Consider a deformed operator
$\tilde{a} = \frac{1}{\sqrt{2}} \left( \lambda  \hat{q} + \frac{i}{\lambda} \hat{p} \right) = \frac{1}{2} \left( \lambda + \frac{1}{\lambda} \right) a + \frac{1}{2} \left(\lambda - \frac{1}{\lambda} \right) a^+$,  where  $\lambda$ (real, positive) is the  squeezing parameter.
It follows that $\tilde{a} = \cosh(v) a + \sinh(v) a^{+}$. 
We now define the squeezed vacuum by $\tilde{a} \lvert 0;\lambda \rangle =0$.
By Robertson-Schrodinger, (see appendix 1),  
this is a minimum uncertainty state. Define translates of 
the squeezed vacuum as $\lvert \alpha; \lambda \rangle = D(\alpha) \lvert  0 ; \lambda \rangle $.
Then $D^{-1}(\alpha) \tilde{a} D(\alpha) = \tilde{a} + \alpha \cosh(v) + \bar{\alpha} \sinh(v) $. 
It follows that $\tilde{a} D(\alpha) \lvert 0;\lambda \rangle = \tilde{\alpha} D(\alpha)\lvert 0;\lambda \rangle $
where $\tilde{\alpha} = \alpha \cosh(v) + \bar{\alpha} \sinh(v) = e^v \alpha_1 + i e^{-v} \alpha_2$ where $\alpha = \alpha_1 + i \alpha_2$. We write $
 D(\alpha) \lvert 0;\lambda \rangle$ as  
$ \lvert \tilde{\alpha};\lambda \rangle$.
Below we drop the tilde and just write $ \lvert \alpha;\lambda \rangle$ for 
the squeezed states. 

Fix $\lambda$ and consider the embedding of the complex plane into the ray space:
$\alpha \mapsto \lvert \alpha; \lambda \rangle  \mapsto [   \lvert \alpha; \lambda \rangle  ].$
The translation group acts on the complex plane which is a homogenous space. To calculate the pull back of the Fubini-Study form by the above embedding,  it is enough to work at the origin.  
Consider $\alpha = \alpha(s)$, a curve in the complex plane such that $\alpha(0) =0$ and $\frac{d \alpha(s)}{d s} |_{s=0}= \alpha^{\prime} $. 
$\lvert \alpha(s) ; \lambda \rangle = D(\alpha(s)) \lvert 0; \lambda \rangle $.
Let $\lvert \psi' \rangle = \frac{d \lvert \alpha(s); \lambda \rangle }{ds} = \frac{d D(\alpha(s)) \lvert 0; \lambda \rangle }{ds}.$ From 
 $\frac{d D(\alpha(s)) }{ds}|_{s=0} = \frac{d \alpha}{ds} a^+ - \frac{d \bar{\alpha}}{ds} a $, we find that
$\lvert \psi^{\prime}  \rangle =\left( \alpha^{\prime} a^+ - \bar{\alpha}^{\prime} a \right) \lvert 0; \lambda \rangle $.
Likewise  for a curve parametrised by $t$, $\lvert \Dot{\psi} \rangle = \left( \Dot{\alpha} a^+ - \Dot{\bar{\alpha}} a \right) \lvert 0; \lambda \rangle $ where $\cdot$ stands for $\frac{d}{dt}$. 
With $\alpha_1, \alpha_2$ the real and imaginary parts of $\alpha$,
$\lvert \psi^{\prime}  \rangle = \left(\alpha_1^{\prime} e^{v} + i \alpha^{\prime}_2 e^{-v} \right) \tilde{a}^+ \lvert 0; \lambda \rangle $ and 
$\lvert\Dot{ \psi }\rangle = \left(\Dot{\alpha_1} e^v + i \Dot{\alpha_2} e^{-v} \right) \tilde{a}^+ \lvert 0; \lambda \rangle $.

To evaluate the pullback of the K\"{a}hler form, 
we need to find $\langle \Dot{\psi} \rvert P \lvert \psi^{\prime} 
\rangle$ where $P = \mathds{1} - \lvert 0; \lambda \rangle 
\langle 0 ; \lambda \rvert $. 
$\langle \Dot{\psi}  \rvert P \lvert \psi^{\prime} \rangle  = \langle 0; 
\lambda \rvert \tilde{a} \tilde{a}^+ \lvert 0; \lambda \rangle  - \langle \Dot{\phi} \lvert 0; \lambda \rangle \langle 0 ; \lambda \rvert \psi^{\prime} \rangle .$ 
Using $[ \tilde{a}, \tilde{a}^+ ] = 1 $ we have (since $\langle 0 ; \lambda \rvert \psi^{\prime} \rangle  =0$)
$\langle \Dot{\phi} \rvert \psi^{\prime} \rangle  = 
(\alpha_1^{\prime} \Dot{\alpha_1} e^{2v} + \alpha_2^{\prime} \Dot{\alpha_2} e^{-2v} ) + i ( \Dot{\alpha_1} \alpha_2^{\prime} - \Dot{\alpha_2} \alpha_1^{\prime} ).$
The pullback of the Riemannian metric (the real part of the above expression)  
agrees with the Riemannican metric on the complex plane only for $v=0$. 
However, the pullback of the symplectic form (the imaginary part of the above
expression) does agree with the symplectic form on the complex plane for all $v$.  
Thus the pull back by squeezed states gives a 
symplectic structure but not a K\"ahler structure.

\subsection{ Symplectic  and K\"{a}hler embedding of the sphere in Projective Hilbert space}

Let the generators of the Lie algebra of $\mathrm{SU(2)}$ be $L_x, L_y, L_z$, all Hermitian. Then the non-vanishing commutators are 
$[L_x, L_y] = i L_z$, $[L_y, L_z] = i L_x$  and $[L_z, L_x] = i L_y$. Let $L_+ = (L_x + i L_y)/\sqrt{2}$ and $L_{-} = (L_x - i L_y)/\sqrt{2}$.  
Then $[L_z, L_{+} ] = L_{+} $, $[L_z, L_{-}] = -L_{-}$ and 
$[L_{+}, L_{-}] =  L_z$. 
With the parameter $\lambda = e^v$, 
let us define ``squeezed operators'' 
$\tilde{L}_{\pm} = e^v L_x \pm i e^{-v} L_y$.

In analogy with the earlier case, we define
the squeezed vacuum state as the kernel of $\tilde{L}_-$
\begin{equation}
\tilde{L}_- \lvert 0; \lambda \rangle =0.  
\end{equation}

Let $\alpha\in {\mathbb C}$ be the  stereographic coordinate of a point on the sphere. 
Then $D(\alpha) = \mathrm{exp}( \bar{\alpha} L_- - \alpha L_+ )$ 
is the displacement operator \cite{perelomov}. A general squeezed state can be got by displacing
the squeezed vacuum using  $D(\alpha)$. The squeezed 
operators $\tilde{L}$ are related to the $L$s as follows:
\begin{eqnarray}
\tilde{L}_+ &=& \cosh(v) L_+ + \sinh(v) L_- ,\\
\tilde{L}_- &=& \sinh(v) L_+ + \cosh(v) L_-.
\end{eqnarray}
Or in inverse form,
\begin{eqnarray}
L_+ &=& \cosh(v) \tilde{L}_+ - \sinh(v) \tilde{L}_-,\\ 
L_- &=& -\sinh(v) \tilde{L}_+ + \cosh(v) \tilde{L}_-.
\end{eqnarray}

Let us now evaluate the pullback of the K\"{a}hler form on the projective
Hilbert space. Because the coadjoint orbit is a homogeneous space, it is sufficient to do the calculation at a single point, which we conveniently choose to be the squeezed vacuum. 
Let $\alpha=\alpha(s)$ be a curve in the complex plane such that $\alpha(0) =0$. 
Then $ \lvert \alpha(s); \lambda \rangle = D(\alpha)\lvert 0; \lambda \rangle$, where as before  
$\lvert 0; \lambda \rangle$ is the squeezed vaccum and 
$\lvert \alpha(s) ; \lambda \rangle$ is the displaced squeezed state.
Then one sees that
$\lvert \psi^{\prime} \rangle = \frac{d}{ds} \lvert \alpha; 
\lambda \rangle = (\frac{d \bar{\alpha}}{ds} L_- - \frac{d \alpha }{ds} L_+) 
\lvert 0; \lambda \rangle = -(\bar{\alpha}^{\prime} \sinh(v) + \alpha^{\prime} 
\cosh(v)) \tilde{L}_+ \lvert 0; \lambda \rangle $.

Thus
$\lvert \psi^{\prime} \rangle = (\alpha_1^{\prime} e^{v} + i \alpha_2^{\prime} e^{-v}) \tilde{L}_+ \lvert 0; \lambda \rangle $, where $\alpha = \alpha_1 + i \alpha_2$.
Similarly,  
$\lvert \Dot{\psi}  \rangle = (\Dot{\alpha_1} e^{v} - i \Dot{\alpha_2}e^{-v})    \tilde{L}_- \lvert 0; \lambda \rangle $.
We calculate $\langle \Dot{\psi} \lvert P \rvert \psi^{\prime} \rangle $ where $P = \mathds{1} - \lvert 0;\lambda \rangle \langle 0; \lambda \rvert $. 
Since $\langle 0; \lambda \rvert \psi^{\prime} \rangle = 0 $ and  
$\langle \Dot{\psi} \rvert \psi^{\prime} \rangle =  
(\alpha_1^{\prime} e^{v} + i \alpha_2^{\prime} e^{-v}) 
(\Dot{\alpha_1} e^{v} - i \Dot{\alpha_2}e^{-v}) \langle 0; 
\lambda | \tilde{L}_- \tilde{L}_+ \rvert 0; \lambda 
\rangle $ and $[\tilde{L}_-, \tilde{L}_+ ] = -  L_z$, we find 
that the pull back Hermitian metric is (apart from an overall constant $- \langle 0; \lambda \rvert L_z \rvert 0; \lambda \rangle$),
\begin{equation}
\left\{ (\alpha_1^{\prime} \Dot{\alpha}_1 e^{-2v} + \alpha_2^{\prime} \Dot{\alpha}_2 e^{2v} ) 
+ i ( \alpha_2^{\prime} \Dot{\alpha}_1 - \alpha^{\prime}_1 \Dot{\alpha_2} ) \right\}. 
\end{equation}
The real part corresponds to the Riemannian metric on the complex plane and the imaginary part corresponds to the symplectic form.
For $v=0$ it is a K\"ahler embedding,  
but if $ v \neq 0$, we get only a symplectic 
embedding since the Riemannian metric is 
not compatible with the symplectic form and the complex structure.

\subsection{K\"ahler embedding of the unit disc: $\mathrm{SU}(1,1)$ case}

This following calculation repeats for 
the upper half plane the calculation we did for Schr\"{o}dinger coherent states. The argument is very similar.  It is actually more convenient to work on the unit disc. The expression for coherent states are 
taken from section 2 of \cite{brif}. There is a 
parameter $k$ which is continuous and 
$k>\frac{1}{2}$.
\newline

The coherent states are given as: 
$$\ket{\alpha,k}=(1-\overline{\alpha}\alpha)^k \sum_{n=0}^\infty \left[ \frac{\Gamma(n+2k)}{n!\Gamma(2k)} \right]^{\frac{1}{2}}\alpha^n\ket{n,k},$$
where $\bra{m,k}{n,k} \rangle=\delta_{mn}$. $k$ is fixed from this point on. \newline
Consider a curve $\alpha(s)$ in the unit disc. The tangent vector is given by: $$\alpha'(s)=\frac{d}{ds}\alpha(s).$$
The corresponding tangent vector in $\mathcal{N}$ is the pushforward, 
$$\frac{d}{ds}\ket{\alpha(s),k}=u\ket{\alpha(s),k}+(1-\overline{\alpha}\alpha)^k \sum_{n=0}^\infty \left[ \frac{\Gamma(n+2k)}{n!\Gamma(2k)} \right]^{\frac{1}{2}}n\alpha^{n-1}\ket{n,k}\alpha'.$$
Now, we have the projection operator, $P=\mathds{1}-\ket{\alpha}\bra{\alpha}$. Denote:
\begin{equation*}
\begin{split}
\ket{b}&=P\frac{d}{ds}\ket{\alpha(s),k}\\
&=(1-\overline{\alpha}\alpha)^k \sum_{n=0}^\infty 
\left[ \frac{\Gamma(n+2k)}{n!\Gamma(2k)} \right]^{\frac{1}{2}}n\alpha^{n-1}\ket{n,k}\alpha',\\
\end{split}
\end{equation*}

\begin{equation*}
\begin{split}
\ket{a}&=P\frac{d}{dt}\ket{\alpha(t),k}\\
&=(1-\overline{\alpha}\alpha)^k \sum_{n=0}^\infty \left[ \frac{\Gamma(n+2k)}{n!\Gamma(2k)} \right]^{\frac{1}{2}}n\alpha^{n-1}\ket{n,k}\Dot{\alpha}.\\
\end{split}
\end{equation*}
\newline
We need to compute $$\langle a\rvert P\rvert b\rangle=\langle a \rvert b \rangle - \langle a\rvert\alpha\rangle\langle \alpha\rvert b\rangle.$$
\newline
The terms are calculated using binomial expansion. 
$$\frac{1}{(1-x)^{2k}}=\sum_{n=0}^\infty {2k+n-1 \choose n} x^n=\sum_{n=0}^\infty \frac{\Gamma(n+2k)}{n!\Gamma(2k)}x^n.$$
\newline
Now,
\begin{flalign*}
\langle a\rvert b\rangle &=(1-\overline{\alpha}\alpha)^{2k}\overline{\alpha'}\dot{\alpha}\sum_{n=0}^\infty \frac{\Gamma(n+2k)}{n!\Gamma(2k)}n^2(\overline{\alpha}\alpha)^{n-1}&\\
   &=\frac{2k}{(1-\overline{\alpha}\alpha)^2}\overline{\alpha'}\dot{\alpha}(1+2k\overline{\alpha}\alpha),&
\end{flalign*}

\begin{flalign*}
        \langle a\rvert\alpha\rangle&=(1-\overline{\alpha}\alpha)^{2k}\overline{\alpha}\dot{\alpha}\sum_{n=0}^\infty \frac{\Gamma(n+2k)}{n!\Gamma(2k)}n(\overline{\alpha}\alpha)^{n-1}&\\
        &=\frac{2k}{1-\overline{\alpha}\alpha}\overline{\alpha}\dot{\alpha} ,&
\end{flalign*}

\begin{flalign*}
\langle \alpha \rvert b \rangle &= \frac{2k}{1-\overline{\alpha}\alpha}\overline{\alpha'}\alpha.&\\
\end{flalign*}
Putting it all together, we calculate:
\begin{equation*}
    \begin{split}
        \langle a \rvert b \rangle - \langle a\rvert\alpha\rangle\langle \alpha\rvert b\rangle&=\overline{\alpha'}\dot{\alpha}\left[\frac{2k}{(1-\overline{\alpha}\alpha)^2}(1+2k\overline{\alpha}\alpha)-\frac{(2k)^2\overline{\alpha}\alpha}{(1-\overline{\alpha}\alpha)^2}\right]\\
        &=\frac{2k\overline{\alpha'}\dot{\alpha}}{(1-\overline{\alpha}\alpha)^2},
    \end{split}
\end{equation*}
which is the K\"{a}hler metric on the unit disc apart from a constant 
depending on $k$. 

{\bf Note:} 
The calculations above for $\mathrm{SU(1,1)}$ 
and those for $\mathrm{SU(2)}$ given above can be adapted to each other after supplying 
the required signs. We do not repeat the calculations for the squeezed states for $\mathrm{SU(1,1)}$.

\section{An Example: Upper Triangular Matrices 
of unit determinant}

In this section we study the geometric quantisation
of the coadjoint orbits of the group of $2 \times 2$ upper triangular matrices with a   different point of view from \cite{DG}.  In the process we give elementary proofs of some known results in \cite{Ad}.

Let $\mathrm{SUT}(2,\mathbb{R})$ denote the Lie group 
of $2 \times 2$ upper triangular matrices over the reals with determinant $1$.
$$\rm{SUT}(2,\mathbb{R})=\left\{ \begin{pmatrix}
a & b \\
0 & \frac{1}{a}\\
\end{pmatrix}: \it{a}  \in \mathbb{R} \backslash \{\rm{0}\},      \it{b}  \in \mathbb{R} \right  \}.$$
$\rm{SUT}(2,\mathbb{R})$ is not a connected group. It has two disjoint components, matrices of positive diagonal and negative diagonal elements respectively. Let $\mathfrak{sut}$ denoted the Lie algebra and $\mathfrak{sut}^*$ denote it's dual. Then
$$\mathfrak{sut}=\left \{ \begin{pmatrix}
u & v\\
0 & -u \\
\end{pmatrix}: u,v \in \mathbb{R} \right \},$$
$$\mathfrak{sut}^*=\left \{ \begin{pmatrix}
u & 0\\
v & -u \\
\end{pmatrix}: u,v \in \mathbb{R} \right \}.$$
The group $\mathrm{SUT}(2,\mathbb{R})$ acts on its Lie algebra $\mathfrak{sut}$ by the adjoint action. This induces an action on the dual space, $\mathfrak{sut}^*$, called the co-adjoint action. 
Take the following basis for $\mathfrak{sut}$, $\mathcal{B}=\{E_1, E_2\}$

\begin{equation*}
    \begin{split}
        E_1 & =\begin{pmatrix}
                   0 & 1 \\
                   0 & 0 \\
                  \end{pmatrix} , \\
E_2& =\begin{pmatrix}
1 & 0 \\
0 & -1 \\
\end{pmatrix}.
    \end{split}
\end{equation*}
The corresponding dual basis of $\mathfrak{sut}^*$ is given by $\mathcal{B}^*=\{E_1^*,E_2^*\}$
\begin{equation*}
    \begin{split}
        E_1^* & =\begin{pmatrix}
                   0 & 0 \\
                   1 & 0 \\
                  \end{pmatrix} ,\\
E_2^*& =\frac{1}{2}\begin{pmatrix}
1 & 0 \\
0 & -1 \\
\end{pmatrix}.
    \end{split}
\end{equation*}

The factor of $\frac{1}{2}$ in $E_2^*$ comes from the the normalization  $\rm{Tr}(E_2 E_2^*)=1$.
 
Let $g=\begin{bmatrix}
g_1 & g_2 \\
0 & \frac{1}{g_1} \\
\end{bmatrix}$ be an element in $\mathrm{SUT}(2,\mathbb{R})$ and $X=\begin{bmatrix}
u_0 & 0 \\
v_0 & -u_0 \\
\end{bmatrix}$ be an element of  $\mathfrak{sut}^*$. 

\begin{lem}
The co-adjoint action of $g \in \mathrm{SUT} (2,\mathbb{R})$ on $\mathfrak{sut}^*$, denoted by ${\rm Ad}_g^*$, is given by the following:
$$\mathrm{Ad}_g^*:\mathfrak{sut}^* \to \mathfrak{sut}^*$$

such that 
 $$\mathrm{Ad}_g^* \begin{bmatrix}
u_0 & 0 \\
v_0 & -u_0 \\ 
\end{bmatrix} =   \begin{bmatrix}
u_0 + g_2 g_1^{-1} v_0 & 0 \\
v_0 g_1^{-2} & -(u_0  + g_2 g_1^{-1} v_0)  \\ 
\end{bmatrix}     .$$
\end{lem}

\begin{proof}
Let us denote by $\langle A,B \rangle = \mathrm{Tr}(AB)$, $A \in \mathfrak{sut}^*$, $B \in \mathfrak{sut}$.
By defintion, $\langle \mathrm{Ad}g^* X, Y \rangle = \langle X, \mathrm{Ad}_g Y \rangle$. where $\mathrm{Ad}_{g ^{-1}}Y = g^{-1} Y g$.

It is easy to see that if $Y=E_1$ $\mathrm{Ad}_{g^{-1}} Y = g_{1}^{-2} E_1$. 
 
Similarly  if $Y =E_2$,  $\mathrm{Ad}_{g^{-1}} Y = 2 g_2 g_1^{-1} E_1 +E_2$. 
 
Then if $Y =E_1$,  $ \langle  \mathrm{Ad}_g^* (2 u_0 E_2^* + v_0 E_1^*),  Y \rangle = v_0 g_{1}^{-2}$. 
 
If $Y = E_2$,  then  $\langle \mathrm{Ad}_g^* (2 u_0 E_2^* + v_0 E_1^*), Y \rangle = 2 (u_0 + v_0 g_2 g_1^{-1} )$.

Let $\mathrm{Ad}_g^*  \begin{bmatrix}
u_0 & 0 \\
v_0 & -u_0 \\ 
\end{bmatrix} =   \begin{bmatrix}
a & 0 \\
c & -a  \\ 
\end{bmatrix}$. Then clearly, $\mathrm{Tr} \left(   \begin{bmatrix}
a & 0 \\
c & -a  \\ 
\end{bmatrix}\begin{bmatrix}
0& 1 \\
0 & 0  \\ 
\end{bmatrix} \right) = c= v_0 g_1^{-2}$.

Similarly,     $\mathrm{Tr} \left(   \begin{bmatrix}
a & 0 \\
c & -a  \\ 
\end{bmatrix}\begin{bmatrix}
1& 0 \\
0 & -1  \\ 
\end{bmatrix} \right) = 2a= 2 ( u_0 + v_0 g_2 g_1^{-1} )$   and  thus $a= u_0 + v_0 g_2 g_1^{-1}$. 
\end{proof}
\begin{lem}
The orbit of the co-adjoint action on $X \in \mathfrak{g}^*$, denoted by $\mathcal{O}_X$ is given as follows.
\begin{equation*}
    \begin{split}
        \mathcal{O}_X &= \{\mathrm{Ad}_g^*.X: g \in \mathrm{SUT}(2,\mathbb{R})\} \\
        &=\left\{\begin{bmatrix}
        u_0-\lambda v_0 & 0 \\
        \mu v_0 & -u_0+\lambda v_0 \\
        \end{bmatrix} : \lambda, \mu \in \mathbb{R},\; \mu > 0
        \right \}.
    \end{split}
\end{equation*}
$\mathcal{O}_X$ is a 2-plane when $v_0 \neq 0$   (homeomorphic to the upper half plane).  It is a point when $v_0=0$.
\end{lem}

\begin{proof}
Take $\mu = g_1^{-2}$ and $\lambda =- g_2 g_1^{-1}$. 
\end{proof}
From now on, we take $v_0 \neq 0$.\newline
Fix a point $A=\begin{bmatrix}
a_1 & 0 \\
a_2 & -a_1 \\
\end{bmatrix}$ in $\mathcal{O}_X$. 
$\mathcal{O}_X$ is 
the upper half plane and any point on this plane is given as $\begin{bmatrix}
s & 0\\
t & -s \\
\end{bmatrix}$. This is diffeomorphic to $\mathbb{R}^2$ with global charts given as :

$\Psi_{v_0}:  \begin{bmatrix}
s & 0\\
t & -s \\
\end{bmatrix} \mapsto (\lambda, \mu) $, where $\lambda = \frac{u_0 - s}{v_0}$ and $ \mu = \frac{t}{v_0}$. 

Let $v_0 >0$.  The upper half plane has a symplectic (K\"ahler) form given by $d \lambda \wedge \frac{d \mu}{\mu^2}$.  The pullback form on $\mathcal{O}_X$ is given by $\Psi_{v_0}^* ( d \lambda \wedge \frac{d \mu}{\mu^2}) =  \frac{dt}{t^2} \wedge ds $. We will use this later.

The tangent space at $A$ is denoted by $T_A \mathcal{O}_X$. With the chart above, elements of the tangent space $T_A \mathcal{O}_X$ are given as $\mu_1 \left.\frac{\partial}{\partial s}\right|_{A}+ \mu_2 \left.\frac{\partial}{\partial t}\right|_A$. 

\begin{lem}\label{id}
Every element of $\xi \in T_A \mathcal{O}_X$ has a representation given as $\xi = {\mathrm ad}_V^*A$ (where $V \in \mathfrak{sut}$ and $\mathrm{ad}_V^*$ is the differential of $\mathrm{Ad}^*$ at $V$)
\end{lem}
\begin{proof}
$$\mathrm{Ad}^*: \mathrm{SUT}(2,\mathbb{R}) \to \mathrm{End}(\mathfrak{sut}^*).$$
$$d(\mathrm{ Ad}^*)=\mathrm{ad}^*:\mathfrak{sut} \to \mathrm{End}(\mathfrak{sut}^*).$$
Denote, $\mathrm{ad}^*(V)= \mathrm{ad}_V^*$. Now take $\mu \in \mathfrak{sut}^*$ and $Y \in \mathrm{SUT} (2,\mathbb{R})$.
\begin{equation*}
    \begin{split}
        \langle \mathrm{ad}_V^*\mu,Y\rangle & = \langle \mu, -\rm{ad}_V(Y) \rangle \\
        & = -\langle \mu,[V,Y] \rangle. 
    \end{split}
\end{equation*} 
\\
Consider the following curve:
\begin{equation*}
    \begin{split}
        \gamma & : \mathbb{R} \to \mathcal{O}_X \\
        & t \mapsto \mathrm{Ad}^*_{e^{tV}}(A).
    \end{split}
\end{equation*}
\begin{equation*}
    \begin{split}
        \gamma(0)&=\mathrm{Ad}^*_{e^{0.V}}(A) \\
        &=\mathrm{Ad}_I^*(A)\\
        &=A \\
        \gamma^{'}(0)&=\mathrm{ad}^*_V A.
    \end{split}
\end{equation*}
Since $\gamma$ is a smooth curve in $\mathcal{O}_X$ with $\gamma(0)=A$, $\gamma^{'}(0) \in T_A\mathcal{O}_X$, we deduce that $\rm{ad}^*_V A \in T_A\mathcal{O}_X$.
Conversely, given $Y \in T_A\mathcal{O}_X$, there exists $V^{'} \in \mathfrak{g}$ s.t $\rm{ad}^*_{V^{'}}(A)=Y$. We make the identification more precise. 
Let $V=\begin{bmatrix}
v_1 & v_2 \\
0 & -v_1 \\
\end{bmatrix}.$
\begin{equation*}
    \begin{split}
        \langle \mathrm{ad}_V^*(A),E_i \rangle & = - \langle A,[V,E_i] \rangle \\
        &=-\mathrm{tr}(A.[V,E_i]).
    \end{split}
\end{equation*}

       $$ \mathrm{ad}_V^*(A)=\sum_{i=1}^2 \langle \mathrm{ad}_V^*(A),E_i \rangle E_i^*. $$
The identification of the two definition of tangent spaces is given as $E_1^* \mapsto -\frac{\partial}{\partial t}$ and $E_2^* \mapsto -\frac{\partial}{\partial s}$.
 Hence$$ \sum_{i=1}^2 \langle \mathrm{ad}_V^*(A),E_i \rangle E_i^* \mapsto 2a_2 v_2 \left.\frac{\partial}{\partial s}\right |_{A} + 2a_2 v_1 \left.\frac{\partial}{\partial t}\right |_A.
$$
\end{proof}

We have the following well known symplectic form.
\begin{lem}
Let $\xi_1,\xi_2 \in T_A\mathcal{O}_X$. Then from lemma \ref{id} we know we can write, $$\xi_1=\rm{ad}_{V_1}^*A,$$ $$\xi_2=\mathrm{ad}_{V_2}^*A$$
where $V_1,V_2 \in \mathfrak{g}$. The Kirillov-Kostant-Souriau symplectic form, $\omega \in \Lambda^2(\mathcal{O}_X)$ is defined as follows:
\begin{equation*}
\begin{split}
    \omega_A(\xi_1,\xi_2):&=\omega_A(\rm{ad}_{V_1}^*A,  \mathrm{ad}_{V_2}^*A) \\
                          &= \langle A,[V_1,V_2] \rangle \\
                          & = A([V_1,V_2]).
\end{split}
\end{equation*}
\end{lem}
 Then we have the following well known theorem (see \cite{Ad}, \cite{DG}  for instance) for which we give a simple proof.
 
\begin{thm}\label{symform}[$\mathcal{O}_X$ is a symplectic manifold]
For $v_0 \neq 0$, $\mathcal{O}_X$ has a natural structure of a symplectic manifold given by the Kirillov-Kostant-Souriau symplectic form, $\omega$.
\end{thm}
\begin{proof}
Since $\omega$ is a symplectic form on $\mathcal{O}_X$, $\omega_A \in \Lambda^2(T_A \mathcal{O}_X)$. 
 \\
Now using the definition of tangent space, since the manifold is two dimensional,
\begin{eqnarray*}
    T_A(\mathcal{O}_X)& =\{ (A,\mu_1,\mu_2): \mu_1,\mu_2 \in \mathbb{R}\} \\
        & = \{\mathrm{ad}_V^*(A): V\in \mathfrak{sut}\}.
    \end{eqnarray*} 
Take two tangent vectors $\mu,\xi$.
\begin{equation*}
    \begin{split}
        \mu &= \mu_1 \left.\frac{\partial}{\partial s}\right|_A+ \mu_2 \left. {\frac{\partial}{\partial t}}\right|_{A} , \\
        \xi &= \xi_1 \left.\frac{\partial}{\partial s}\right |_A+ \xi_2 \left.\frac{\partial}{\partial t}\right|_A .\\
    \end{split}
\end{equation*}
We calculate $X_\mu, X_\xi$ such that
\begin{equation*}
    \begin{split}
        \mu&=\mathrm{ad}^*_{X_\mu}(A)= \begin{bmatrix}
        \mu_1 & 0 \\
        -\mu_2 & -\mu_1 \\
        \end{bmatrix}, \\
     \xi&=\mathrm{ad}^*_{X_\xi}(A)=\begin{bmatrix}
        \xi_1 & 0 \\
        -\xi_2 & -\xi_1 \\
        \end{bmatrix}.
    \end{split}
\end{equation*}
Using the identification given above, 
\begin{equation*}
    \begin{split}
        X_\mu&=\begin{bmatrix}
        \frac{\mu_2}{2a_2} & \frac{\mu_1}{2a_2}\\
        0 & -\frac{\mu_2}{2a_2}\\
        \end{bmatrix} , \\
        X_\xi&=\begin{bmatrix}
        \frac{\xi_2}{2a_2} & \frac{\xi_1}{2a_2}\\
        0 & - \frac{\xi_2}{2a_2}\\
        \end{bmatrix}.
    \end{split}
\end{equation*}
Now we calculate the symplectic form:
\begin{equation*}
    \begin{split}
    \omega_A(\mu,\xi)&=\omega_A \left(\rm{ad}^*_{X_\mu}(A), \rm{ad}^*_{X_\xi}(A) \right) \\
    &= \langle A,[X_\mu,X_\xi] \rangle \\
    &= -\mathrm{Tr}\left(A.[X_\mu,X_\xi] \right) \\
    &= -\frac{1}{2a_2}(\xi_1\mu_2-\mu_1\xi_2).
    \end{split}
\end{equation*}
Hence,$$\omega_A=\frac{1}{a_2}ds|_A \wedge dt|_A,$$ 
$$\omega=\frac{1}{t}ds\wedge dt.$$
\end{proof}

Now,
let $P=\begin{bmatrix}
s & 0 \\
t & -s \\
\end{bmatrix} \in \mathcal{O}_X$. Define the functions $J_1,J_2$ on $\mathcal{O}_X$ as follows: 
$$J_i=\mathrm{tr}(P E_i).$$ Then $J_1 = t$ and $J_2 = 2s$.

Since $\omega = ds \wedge \frac{dt}{t}$, we get 
\begin{lem}
The Hamiltonian vectors fields corresponding to $J_i$'s are denoted by $X_{J_i}$, satisfying, $\omega(X_{J_i},-)=dJ_i(-)$ are
\begin{equation*}
    \begin{split}
        X_{J_1} & = t\frac{\partial}{\partial s}, \\
        X_{J_2} & = -2t\frac{\partial}{\partial t}. \\
    \end{split}
\end{equation*}
\end{lem}

$\mathcal{O}_X$ is a symplectic manifold and we can  perform geometric prequantization on it.  Denote the Hilbert space of square integrable sections as $\mathcal{H}$. Let $\psi \in \mathcal{H}.$

Let a choice of a  symplectic potential be $\theta = - |\mathrm{log}( t)| ds$.
 
\begin{prop}\label{operators}
The operators corresponding to the $J_i$'s on the Hilbert space after geometric prequantization are denoted by $\hat{J_i}$ and satisfy, $\hat{J_i}\psi=-i\hbar \nabla_{X_{J_i}}+J_i\psi$,  definition as in ~\cite{W}. 
\begin{equation*}
    \begin{split}
      \hat{J_1}\psi & = -i \hbar t \frac{\partial}{\partial s}\psi +t\mathrm{log}|t|. \psi+t\psi,\\
         \hat{J_2}\psi & = 2i \hbar t \frac{\partial}{\partial t}\psi + 2s\psi.\\
    \end{split}
\end{equation*}
\end{prop}

If we exponentiate these we get the operators corresponding to the group. Take $a=\mathrm{log}(t)$,  then 
$$\frac{\partial}{\partial a}=t\frac{\partial }{\partial t}.$$
Note that $s$ and $\frac{\partial}{\partial a}$ commute. Hence we can use the property of exponentials to obtain:

\begin{prop}
\begin{equation*}
    \begin{split}
        e^{\tau J_1}\psi(a,s)&=e^{\tau t(\mathrm{log}{|t|}+1)}\psi(a,s- i\hbar t \tau) ,\\
        e^{\tau J_2}\psi(a,s)&=e^{s\tau}\psi(a+2i\hbar \tau,s). \\
     \end{split}
\end{equation*}
\end{prop}

We have the following inclusion of $\mathcal{O}_X$ in $\mathbb{R}^2$:
$$\begin{pmatrix}
u & 0 \\
v & -u \\
\end{pmatrix} \mapsto (u,v).$$
Then the corresponding action of $\mathrm{SUT}(2,\mathbb{R})$ on $\mathbb{R}^2$ is $(u_0,v_0) \mapsto (u_0-\lambda v_0, \mu v_0)$.
\begin{lem}
The stabilizer of the action, $H={\pm I}$.
\end{lem}
\begin{proof}
\begin{equation*}
\begin{split}
\begin{pmatrix}
u_0 & 0 \\
v_0 & -u_0 \\
\end{pmatrix} &=\begin{pmatrix}
    g_1 & g_2\\
    0 & \frac{1}{g_1}\\
    \end{pmatrix}.\begin{pmatrix}
u_0 & 0 \\
v_0 & -u_0 \\
\end{pmatrix}\\
&=\begin{pmatrix}
u_0-\frac{g_2}{g_1}v_0 & 0 \\
\frac{1}{g_1^2}v_0 & -u_0 + \frac{g_2}{g_1}v_0 \\
\end{pmatrix}.
\end{split}
\end{equation*}
Hence, $\frac{g_2}{g_1}=0$ and $\frac{1}{g_1^2}=1 \implies g_2=0, g_1= \pm 1$. So $H={\pm I}$.
\end{proof}

We work with $\mathrm{SUT}$ matrices with positive diagonals from now on, so that $H$ is trivial. Note that this is one of the connected components. Since $H$ is trivial we have, $\mathrm{SUT}^+ \cong\mathcal{O}_X$.

We then  have  direct proof of  the following  
result in the next proposition. The 
result can be found for instance in \cite{Ad}.

\begin{prop}
$\mathrm{SUT}^+$ is a  K\"ahler manifold.
\end{prop}

\begin{proof}

Let
\begin{equation*}
   \begin{split}
       \Phi_{v_0}:\mathcal{O}_X & \to \rm{SUT}^+ \\
        \begin{pmatrix}
        s & 0 \\
        t & -s \\
        \end{pmatrix} & \mapsto 
        \begin{pmatrix}
        \sqrt{\frac{v_0}{t}} & \frac{u_0-s}{\sqrt{v_0 t}}\\
        0 & \sqrt{\frac{t}{v_0}}\\
        \end{pmatrix}.
   \end{split} 
\end{equation*}

\begin{equation*}
\begin{split}
        \Phi^{-1}_{v_0}:\mathrm{SUT}^+ & \to  \mathcal{O}_X \\
        \begin{pmatrix}
        a & b \\
        0 & \frac{1}{a} \\
        \end{pmatrix} & \mapsto 
        \begin{pmatrix}
        u_0-\frac{b v_0}{a} & 0\\
        \frac{v_0}{a^2} & -u_0+\frac{b v_0}{a} \\
        \end{pmatrix}.
\end{split}
\end{equation*} where $a = \sqrt{\frac{v_0}{t}}$ and $b = \frac{u_0-s}{\sqrt{v_0t}}$.

Recall there is a map $\Psi_{v_0}:  \begin{bmatrix}
s & 0\\
t & -s \\
\end{bmatrix} \mapsto (\lambda, \mu) $, where $\lambda = \frac{u_0 - s}{v_0}$ and $ \mu = \frac{t}{v_0}$. 

Let $\chi_{v_0} = \Psi_{v_0} \circ \Phi_{v_0}^{-1}$ is a homeomorphism of $\mathrm{SUT^{+}}$ to the upper half plane. 

It is given by $\chi_{v_0}: \mathrm{SUT^+} \mapsto \mathrm{H} $, as follows:

\begin{eqnarray*}
& & \begin{pmatrix}
        a & b \\
        0 & \frac{1}{a} \\
        \end{pmatrix}  \mapsto 
         \begin{pmatrix}
        u_0-\frac{b v_0}{a} & 0\\
        \frac{v_0}{a^2} & -u_0+\frac{b v_0}{a} \\
        \end{pmatrix} \\
        & &  \mapsto \left(\frac{b}{a}, \frac{1}{a^2} \right) = (\lambda, \mu) \in \mathrm{H}.
        \end{eqnarray*}

Here $ s = u_0 - \frac{b v_0}{a}$ and $t = \frac{v_0}{a^2}$. 
Then calculating the pull-back of the usual Kahler form $\Omega= d \lambda \wedge \frac{d \mu}{\mu^2}$ on the upper half plane,  we have,  
\begin{equation*}
  \begin{split}
      \chi_{v_0}^*(\Omega) =  \Phi_{v_0}^{-1 *} \circ \Psi_{v_0}^* \left( d \lambda \wedge \frac{d \mu}{\mu^2} \right)  \\
       = \Phi_{v_0}^{-1 *} \left( \frac{dt}{t^2 } \wedge ds \right) \\
         = 2 da \wedge d b = 2 db \wedge \frac{d \tilde{a}}{\tilde{a}^2},
   \end{split}
\end{equation*} where $\tilde{a} = \frac{1}{a}$.
Here we have used the fact that $a = \sqrt{\frac{v_0}{t}}$ and $b = \frac{u_0 -s}{\sqrt{v_0t}}$.

Let the map $F_1$ be as follows.

$F_1: \left(b, \tilde{a} \right) \in \mathrm{H}  \mapsto \begin{pmatrix}
        a & b \\
        0 & \frac{1}{a} \\
        \end{pmatrix} \in \mathrm{SUT^+}  $, 
        where $\tilde{a} = \frac{1}{a}$,  $b + i \tilde{a} \in \mathrm{H}$.   
        
        $F_1^* \circ \chi_{v_0}^* \left( d\lambda \wedge \frac{d\mu}{\mu^2} \right) = 2 db \wedge \frac{d \tilde{a}}{\tilde{a}^2}$ which is a  K\"ahler form on $\mathrm{H}$.
Then it  is easy to see that $\chi^*_{v_0}(\Omega)$ is  K\"ahler via  the map  $F_1$.

\end{proof}

\section{Conclusion}
In this paper we have considered some examples of Lie groups and studied
coherent states as embeddings of the classical phase space into the 
quantum ray space. We find that the states saturating the Heisenberg
uncertainty fall into two classes: coherent states which yield K\"{a}hler
embeddings and squeezed states which only give us symplectic embeddings.
Our motivation was to understand from these examples, which coadjoint 
orbits naturally admit a K\"{a}hler structure and whether this agrees with
the pullback of the structure on projective Hilbert space. Let us 
consider this aspect in more detail.

 Do  co-adjoint orbits of Lie groups 
naturally admit K\"ahler structures? Let us see how our examples 
illuminate this question. In the case of compact semisimple groups,
there is a positive definite Cartan-Killing metric on the Lie algebra. 
This metric permits us to identify $\mathfrak{g}$  with its dual and we get a Riemannian metric on $\mathfrak{g}^*$.
The generic co-adjoint orbit is a submanifold of $\mathfrak{g}^*$ and we get an induced Riemannian metric
on ${\mathcal O}_X$. This leads to a metric and then a K\"ahler structure on ${\mathcal O}_X$.
This is in fact what happens with the $\mathrm{SU(2)}$ example.

Let us now consider noncompact semisimple groups. 
While it is true that the Lie algebra has a non-degenerate
Cartan-Killing metric, this metric has indefinite signature. 
In general the orbits do not necessarily have a Riemannian metric:
the induced metric may be Lorentzian or even degenerate.
It is still possible however, that the orbits inherit a
Riemannian metric. 
This is just what occurs in the case of $\mathrm{SU(1,1)/U(1)}$ that we considered. 
While the signature of the Cartan-Killing form is $\{-,-,+\}$, the $+$ direction
is in the stabiliser $\mathrm{U(1)}$
of the group action on the orbit.
As a result there is a (negative) definite metric 
on the co-adjoint orbit. Reversing the sign of the metric gives us a 
Riemannian metric. This example as well as the previous one fall within 
the framework of ~\cite{V}.

In contrast the Weyl-Heisenberg group is not semi-simple. In fact, the Cartan-Killing form vanishes identically.
We might then suppose that the metric on the complex plane can be derived from the complex structure on the W-H group. 
A complex structure can be combined with the symplectic structure to give us a metric.

The case of $\mathrm{SUT}$ is also similar. 
G is upper triangular and non-compact. The Cartan-Killing form is degenerate (though not identically zero).
We can regard G as a subgroup of $\mathrm{SL}(2,  \mathbb{R})$ where the Cartan-Killing metric is 
nondegenerate, but of Lorentzian signature. This metric does not
pull back to a Riemannian metric on ${\mathcal O}_X$. 
We are led to ask what metrics naturally exist on ${\mathcal O}_X$?
Since ${\mathcal O}_X$ is a homogeneous space $G/H$, we would expect that the metric must also be homogeneous and therefore
have constant curvature. In two dimensions, constant curvature metrics  can only be the plane, the sphere and the upper half plane.
The second is ruled out since it is compact, unlike the coadjoint orbit of 
$\mathrm{SUT}$.
The others have isometry groups which are respectively $\mathrm{E}(2)$ and $\mathrm{SO}(2,1)$. The Lie algebras of these groups would generate
symmetries of ${\mathcal O}_X$. In our case the symmetry group of ${\mathcal O}_X$ is two dimensional and nonabelian. The Lie algebra of $\mathrm{E}(2)$  does not 
admit such a subalgebra. We conclude that the natural Riemannian metric on
 ${\mathcal O}_X$ is the constant negative curvature metric. We can use such 
a metric to determine a complex structure on ${\mathcal O}_X$, thereby turning it into a K\"ahler manifold. 
We have seen that in cases where there is a natural K\"{a}hler structure on the co-adjoint orbit, the technique of pulling 
back using coherent states agrees with the natural structure up to an overall constant.  It may be possible to use this technique more
generally to induce K\"{a}hler structures on co-adjoint orbits.

\section{Appendix}
\subsection{Appendix1: Minimum Uncertainty States:}
This appendix summarises the argument for the Robertson-Schr\"{o}dinger
uncertainty relations \cite{Robertson,schrod} which were used in the text to motivate our definition
of squeezed states in a broader context than SCS. Let $\hat{A}, \hat{B}$ be Hermitian 
operators on a Hilbert space ${\mathcal H}$.  For any fixed, normalized $\rvert\psi\rangle \in {\mathcal H}$, 
define $a= \langle\psi  \rvert \hat{A} \lvert \psi \rangle$ and $b = \langle\psi  \rvert \hat{B} \lvert \psi \rangle$ and $\tilde{A} = \hat{A} - a I$,  $\tilde{B} = \hat{B} - b I$.
Let $\lvert \sigma  \rangle = \left(\tilde{A} + \gamma \tilde{B} \right) \lvert \psi  \rangle$ where $\gamma = \gamma_1 + i \gamma_2$ is a complex number. 
The function $F(\gamma) = \langle \sigma \rvert \sigma \rangle \geq 0$ of $\gamma$ is non-negative and vanishes only when $\lvert \sigma  \rangle$ does.
Then $F(\gamma) = |\gamma|^2 \beta + \gamma_1 \langle \psi | \left\{ \tilde{A}, \tilde{B} \right \} \rvert \psi \rangle  + i  \gamma_2 \langle \psi | \left[ \tilde{A}, \tilde{B} \right] \rvert \psi \rangle + \alpha$,
where $\beta = \langle \psi \rvert \tilde{B}^2 \rvert \psi \rangle$,  $\alpha =\langle \psi \rvert \tilde{A}^2 \rvert \psi \rangle$.
Defining 
$C_+ = \langle \psi | \{ \tilde{A}, \tilde{B} \} \rvert \psi \rangle $ 
and $C_ {-}= i \langle \psi | [ \tilde{A}, \tilde{B} ] \rvert \psi \rangle$, we have $F(\gamma) = \beta |\gamma|^2 + \gamma_1 C_+ + \gamma_2 C_{-} + \alpha \geq 0$.

To derive the uncertainty relations,  we consider $F(\gamma)$ as a quadratic in $\gamma$ on the complex $\gamma$ plane and look at special cases..
\begin{enumerate}
\item{Heisenberg uncertainty relations: } One sets $\gamma_1 =0$. 
Then $\beta \gamma_2^2 + C_{-} \gamma_2 + \alpha \geq 0$ for all  $\gamma_2$ (real).
This implies $C_{-}^2 - 4 \beta \alpha \leq 0$ or $\beta \alpha  \geq \frac{C_{-}^2}{4}$.
For $\hat{A} = \hat{x} $ and $\hat{B} = \hat{p}$ we have $\left[\hat{A},  \hat{B} \right] = i \hbar$ and thus $C_{-} = - \hbar$. 
Then $\beta \alpha \geq \frac{\hbar^2}{4}$. It is easy to see that $\alpha = (\Delta A)^2 $ and $\beta = (\Delta B)^2$ which leads to 

$$\Delta A \Delta B  \geq \frac{\hbar}{2}.$$
\item Set $\gamma_2 = 0$. Then $\beta \gamma_1^2 + \gamma_1 C_{+} + 
\alpha \geq 0$ for all $\gamma_1$ real imples $\beta \alpha \geq C_{+}^2$ or
$$ \Delta A \Delta B \geq \frac{\left| \left\{ \hat{A}, \hat{B} \right \} \right|}{2} .$$

\item Robertson-Schrodinger uncertainty relations: 
The strongest version of the uncertainty relations comes from using complex $\gamma$ and minimising $F(\gamma)$ over the complex plane.  Let $\gamma = \rho e^{i \theta}$.  
Then $F(\rho, \theta) = \beta \rho^2 + \rho \cos(\theta) C_{+} + \rho \sin(\theta) C_{-} + \alpha \geq 0$. 
Setting $\frac{\partial F}{ \partial \rho} =0 $ and $\frac{ \partial F}{ \partial \theta} =0$ we have (since $\rho >0$)
$\tan(\theta) = \frac{C_{-}}{C_{+}}$.  In other words 
$\cos(\theta) = \frac{C_{+} \epsilon}{C}$ and $\sin(\theta) = \frac{C_{-} \epsilon}{C}$ where  $\epsilon = \pm 1$ and $C^2 = C_{+}^2 + C_{-}^2$.  
For the minimum, we compute the second derivatives.  We get $\frac{\partial^2 F}{\partial \rho^2} = 2 \beta$, 
$\frac{\partial^2 F}{\partial \theta^2} = - \rho C \epsilon$ and $\frac{\partial^2 F}{\partial \rho \partial \theta} =0$.
So for  a minimum,  $\epsilon = -1$. 
The value of $F(\rho, \theta) = \frac{-C^2}{4 \beta} + \alpha \geq 0$ at the minimum.
This gives $\alpha \beta \geq \frac{C^2}{4}$  where 
$C^2 = C_{+}^2 + C_{-}^2$ or , 
$$ \Delta A \Delta B \geq \frac{1}{2}  \sqrt{\left|\left\langle  \left\{ \hat{A}, \hat{B} \right\} \right \rangle \right|^2 + \left| \left\langle  \left[\hat{A}, \hat{B} \right] \right\rangle \right|^2 }. $$ 
\end{enumerate}

The Robertson-Schr\"{o}dinger inequality imples the Heisenberg inquality and the inequality in (2) as well.
The inequality is saturated when  
$\rvert \sigma \rangle = ( \tilde{A} + \gamma_0 \tilde{B} ) 
\rvert \psi \rangle =0$. Or with some rearrangement 
(setting $i \lambda^{-2}=\gamma_0$) 
$$\left(\lambda \hat{A} + i \hat{B}/\lambda \right) \rvert \psi \rangle = 
(\lambda a + i b/\lambda ) \rvert \psi \rangle,$$
the condition we use to define minimum uncertainty states. We need only consider
$\lambda$ real and positive.

\subsection{Appendix2: Coherent states and Berezin Quantization of the phase space}
The phase space is identified with the upper half plane $\mathbb{H} = \{ z \in {\mathbb C} : \mathrm{Im}(z) > 0\}$. 
In this section we show the Berezin Quantization of the upper half plane adapting the one for the unit disc as described in Perelomov ~\cite{perelomov}, chapter 16.

Recall there is a biholomorphism from $\mathbb{H}$ to the unit disc $\mathbb{D}$ given by: $\epsilon: \mathbb{H} \rightarrow \mathbb{D}$
$$\epsilon(w) = \frac{w-i}{w+i}.$$
Let $\chi$ be the inverse of $\epsilon$.  The K${\ddot{\rm{a}}}$hler form of $\mathbb{D}$ is 
$d \mu (z, \bar{z}) = \frac{1}{2 \pi i} \frac{d z \wedge d \bar{z}}{(1 - |z|^2)^2} $. Using $z=\epsilon(w)$ we get 
$$d \mu(z, \bar{z}) = \frac{1}{2 \pi i} \frac{d w \wedge d \bar{w}}{ 4( \mathrm{Im} (w))^2}= d \mu(w, \bar{w}).$$
Also, $\left(1-|z|^2 \right)^{1/h} = \left(\frac{4 Im(w)}{|w|^2 + 2 \mathrm{Im} (w) +1 )}\right)^{1/h}$.   
\\
Let $f \in C^{\infty}(\mathbb{H}) $. Then $ f = \epsilon^* (\phi)$ where $\phi \in C^{\infty}(\mathbb{D})$ and $\phi = \chi^*(f)$.
\\
Let $f, g \in C^{\infty}(\mathbb{H})$. Let $ \psi = \chi^*(g)$.  Then $(f,g)_{\mathbb{H}} = (\phi, \psi)_{\mathbb{D}}$ where
\begin{equation}
\begin{split}
(\phi, \psi)_{\mathbb{D}} &= \left(\frac{1}{h} -1\right) \int_{\mathbb{D}} \bar{\phi}(z) \psi(z) (1 - |z|^2) ^\frac{1}{h} d \mu(z,\bar{z})\\
&=\left(\frac{1}{h} -1\right) \int_{\mathbb{H}} \bar{f}(w) g(w)    \left(\frac{4 \mathrm{Im}(w)}{|w|^2 + 2 \mathrm{Im} (w) +1 )}\right)^{1/h} d \mu(w,\bar{w}).\\
\end{split}
\end{equation}
Let ${\mathcal F}_{h}$ be the space of all smooth square integrable functions with respect to the above inner product, i.e. $\|f\|_{\mathbb{H}} < \infty$.
\\
We proceed as in Perelomov, \cite{perelomov}, chapter 16 where the Berezin quantization on the Lobachevsky plane is explained.
\\
Let $\psi_l(z) = (l!)^{-1/2} \times \left[ (\frac{1}{h})....(\frac{1}{h} - 1 +l)\right]^{1/2} \times z^l$ be the orthonormal basis for $\chi^*({\mathcal F}_{h})$ (from 16.3.3, ~\cite{perelomov}).
\newline
Let $f_l(w) = (l!)^{-1/2} \times \left[ (\frac{1}{h})....(\frac{1}{h} - 1 +l)\right]^{1/2} \times \left(\frac{w-i}{w+i}\right)^l$ be a basis for ${\mathcal F}_{h}$.
\newline
$(f_l, f_m)_H = (\psi_l, \psi_m) = \delta_{lm}$. 
\newline
 Let $\tau_p = \sum_{l} \overline{f_l(p)} f_l $ be the coherent state parametrized by p, i.e.
$\tau_p(w) = \sum_{l} \overline{f_l(p)} f_l(w)$, $w \in {\mathbb H}$.
\newline
One can show that for any function $f \in {\mathcal F}_h$, $(\tau_p, f)_{\mathbb{H}} = f(p)$. 
\newline
Let $\hat{{\mathcal P}}$  be a  bounded linear operator acting on ${\mathcal F}_h$. 
\newline
Then it is easy to check that there exists an operator $\hat{A}$ acting on the square integrable smooth functions on $\mathbb{D}$ such that:
$$\hat{{\mathcal P}} (f) = \hat{{\mathcal P}}\left(\epsilon^*(\phi)\right) = \epsilon^*\left(\hat{A}(\phi) \right) = \epsilon^* \hat{A}\left(\chi^*(f) \right),$$
where $\hat{A} = \chi^* \hat{{\mathcal P}} \epsilon^*$.
\newline

The symbol of $\hat{{\mathcal P}}$ is defined to be ${\mathcal P}(p, \bar{q}) = \frac{\left(\tau_p, \hat{{\mathcal P}} \tau_q\right)_{\mathbb{H}}}{
\left(\tau_p, \tau_q\right)_{\mathbb{H}}}$

An easy calculation shows that $\chi^*(\tau_p) = \psi_{\zeta}$ is a coherent state on $\mathbb{D}$ parametrized by $\zeta = \epsilon(p)$.  In fact all coherent states on $\mathbb{D}$ are of this form. Let the symbol of $\hat{A}$ be denoted by $A(\zeta, \bar{\eta})$ where $\eta = \epsilon(q)$. One can  show that the symbol ${\mathcal P}(p, \bar{q}) = A(\zeta, \bar{\eta})$, i.e. the two symbols are in fact the same. 
\newline
\newline
{\bf The star product:} Let $\hat{{\mathcal P}}_1, \hat{{\mathcal P}}_2$ be two bounded linear operators acting on 
${\mathcal F}_h$ and $\hat{A}_1, \hat{A}_2$ be two bounded linear operators acting on $\chi^*({\mathcal F}_h)$.
Then it is easy to show that the star product $\left({\mathcal P}_1 * {\mathcal P}_2\right) (p,\bar{p})= (A_1 * A_2)(\zeta, \bar{\zeta})$ is the symbol for the composition operator $ \hat{{\mathcal P}}_1 \circ \hat{{\mathcal P}}_2$. 

It has been proved in ~\cite{perelomov}, chapter 16, that the correspondence principle holds for $A_1 * A_2$. 
Thus, $$\mathrm{lim}_{h \rightarrow 0} \left({\mathcal P}_1 {*} {\mathcal P}_2 \right) \left(p,\bar{p}\right) = {\mathcal P}_1 (p, \bar{p}) {\mathcal P}_2(p, \bar{p}).$$

Also,
$$\mathrm{lim}_{h \rightarrow 0} \frac{1}{h} \left(\left({\mathcal P}_1 {*} {\mathcal P}_2\right) -  \left({\mathcal P}_2 {*} {\mathcal P}_1 \right)\right)  = \{ A_1, A_2 \}_{\mathbb{D}} = \{ {\mathcal P}_1, {\mathcal P}_2 \}_{\mathbb{H}},$$
 where the Poisson brackets are defined as follows:
\begin{equation}
\begin{split}
\{ A_1, A_2 \}_{\mathbb{D}} &= (1-|z|^2)^2 \left( \frac{\partial A_1}{\partial \bar{z}} \frac{\partial A_2}{\partial z}- \frac{\partial A_2}{\partial \bar{z}}\frac{\partial A_1}{\partial z} \right)\\ 
&=4(\mathrm{Im}(w))^2 \left(\frac{\partial {\mathcal P}_1}{\partial \bar{w}} \frac{\partial {\mathcal P}_2}{\partial w}- \frac{\partial {\mathcal P}_2}{\partial \bar{w}}\frac{\partial {\mathcal P}_1}{\partial w} \right)\\
 &= \left\{{\mathcal P}_1 , {\mathcal P}_2 \right\}_{\mathbb{H}.}\\
\end{split}
\end{equation}
The last equality follows from the fact that $\left(1-|z|^2\right)^2 = \frac{16 \left(\mathrm{Im}(w)\right)^2}{\left(|w|^2 + 2 \mathrm{Im}(w) +1 \right)^2}$ and $ |\frac{\partial w}{\partial z}|^2 = \frac{1}{4} \left(|w|^2 + 2 \mathrm{Im}(w) +1 \right)^2$ where recall $ z= \frac{w-i}{w+i}$. 

\section{Acknowledgement}
R.D. acknowledges support from the Department of Atomic Energy,  Government of India,  under project RTI4001. She also acknowledges support from the grant CRG/2018/002835, Science and Engineering Research Board, Government of India,  
J.S. acknowledges support by a grant from the Simons Foundation (677895, R.G.).
R.S.V.  would like to acknowledge the LTVSP praogram in I.C.T.S.-T.I.F.R. where this project was completed. 
We thank an anonymous referee for critical comments which improved the
manuscript and Supurna Sinha for reading through the
final version.


\end{document}